\def\draft{0}

\documentclass[11pt]{article}

\usepackage{etex}
\usepackage{verbatim}
\usepackage{xspace,enumerate}
\usepackage[dvipsnames]{xcolor}
\usepackage[T1]{fontenc}
\usepackage[full]{textcomp}
\usepackage[american]{babel}
\usepackage{mathtools}
\usepackage{amsthm}
\usepackage{algorithm}
\usepackage[noend]{algpseudocode}
\usepackage{nicefrac}

\usepackage{fullpage}
\usepackage{tcolorbox}
\usepackage{todonotes}

\usepackage{mathpazo}
\usepackage{multicol}
\usepackage{graphicx}
\usepackage{booktabs}

\usepackage[
letterpaper,
top=1in,
bottom=1in,
left=1in,
right=1in]{geometry}
\usepackage{newpxtext} %
\usepackage{textcomp} %
\usepackage[varg,bigdelims]{newpxmath}
\usepackage[scr=rsfso]{mathalfa}%
\usepackage{bm} %
\let\mathbb\varmathbb
\usepackage{microtype}
\usepackage[pagebackref,bookmarksnumbered,colorlinks=true,urlcolor=blue,linkcolor=blue,citecolor=OliveGreen]{hyperref}
\usepackage[capitalise,nameinlink]{cleveref}
\crefname{lemma}{Lemma}{Lemmas}
\crefname{fact}{Fact}{Facts}
\crefname{theorem}{Theorem}{Theorems}
\crefname{corollary}{Corollary}{Corollaries}
\crefname{claim}{Claim}{Claims}
\crefname{example}{Example}{Examples}
\crefname{problem}{Problem}{Problems}
\crefname{definition}{Definition}{Definitions}
\crefname{exercise}{Exercise}{Exercises}
\usepackage{amsthm}

\newtheorem{theorem}{Theorem}[section]
\newtheorem*{theorem*}{Theorem}
\newtheorem{lemma}[theorem]{Lemma}
\newtheorem*{lemma*}{Lemma}

\newtheorem*{fact*}{Fact}

\newtheorem*{proposition*}{Proposition}
\newtheorem{corollary}[theorem]{Corollary}
\newtheorem*{corollary*}{Corollary}

\newtheorem*{hypothesis*}{Hypothesis}

\newtheorem*{conjecture*}{Conjecture}
\theoremstyle{definition}
\newtheorem{definition}[theorem]{Definition}
\newtheorem*{definition*}{Definition}

\newtheorem*{construction*}{Construction}

\newtheorem*{example*}{Example}

\newtheorem*{question*}{Question}

\newtheorem*{assumption*}{Assumption}

\newtheorem*{problem*}{Problem}

\newtheorem*{openquestion*}{Open Question}
\theoremstyle{remark}

\newtheorem*{claim*}{Claim}
\newtheorem{remark}[theorem]{Remark}
\newtheorem*{remark*}{Remark}

\newtheorem*{observation*}{Observation}
\usepackage{paralist}
\let\originalleft\left
\let\originalright\right
\renewcommand{\left}{\mathopen{}\mathclose\bgroup\originalleft}
\renewcommand{\right}{\aftergroup\egroup\originalright}
\usepackage{turnstile}
\usepackage{mdframed}
\usepackage{tikz}
\usetikzlibrary{positioning}
\usepackage{caption}
\usepackage{newfloat}
\usepackage{array}
\usepackage{subfig}
\usepackage{bbm}
\usepackage{xparse}
\usepackage{amsthm} %
\makeatletter
\let\latexparagraph\paragraph
\RenewDocumentCommand{\paragraph}{som}{%
  \IfBooleanTF{#1}
    {\latexparagraph*{#3}}
    {\IfNoValueTF{#2}
       {\latexparagraph{\maybe@addperiod{#3}}}
       {\latexparagraph[#2]{\maybe@addperiod{#3}}}%
  }%
}
\newcommand{\maybe@addperiod}[1]{%
  #1\@addpunct{.}%
}
\makeatother

\usepackage{boxedminipage}

\newcommand{\Paren}[1]{\left(#1\right)}

\newcommand{\brac}[1]{[#1]}
\newcommand{\Brac}[1]{\left[#1\right]}

\newcommand{\card}[1]{\lvert#1\rvert}

\newcommand{\set}[1]{\{#1\}}
\newcommand{\Set}[1]{\left\{#1\right\}}

\newcommand{\Esymb}{\mathbb{E}}
\newcommand{\Psymb}{\mathbb{P}}

\DeclareMathOperator*{\E}{\Esymb}

\DeclareMathOperator*{\ProbOp}{\Psymb}
\renewcommand{\Pr}{\ProbOp}

\RequirePackage[outline]{contour} % used to define a less-bold \xsos
\contourlength{0.065pt}
\contournumber{10}%

\newcommand\bdot\bullet

\DeclareMathOperator{\poly}{poly}

\newcommand{\Erdos}{Erd\H{o}s\xspace}
\newcommand{\Renyi}{R\'enyi\xspace}

\newcommand{\R}{\mathbb R}

\newcommand{\problemmacro}[1]{\texorpdfstring{\textup{\textsc{#1}}}{#1}\xspace}

\newcommand{\maxclique}{\problemmacro{clique}}

\newcommand{\maxbalancedbiclique}{\problemmacro{balanced biclique}}

\newcommand{\cG}{\mathcal G}

\renewcommand{\leq}{\leqslant}
\renewcommand{\le}{\leqslant}
\renewcommand{\geq}{\geqslant}
\renewcommand{\ge}{\geqslant}

\let\epsilon=\varepsilon
\numberwithin{equation}{section}
\newcommand\MYcurrentlabel{xxx}
\newcommand{\MYstore}[2]{%
  \global\expandafter \def \csname MYMEMORY #1 \endcsname{#2}%
}
\newcommand{\MYload}[1]{%
  \csname MYMEMORY #1 \endcsname%
}
\newcommand{\MYnewlabel}[1]{%
  \renewcommand\MYcurrentlabel{#1}%
  \MYoldlabel{#1}%
}
\newcommand{\MYdummylabel}[1]{}
\newcommand{\torestate}[1]{%
  \let\MYoldlabel\label%
  \let\label\MYnewlabel%
  #1%
  \MYstore{\MYcurrentlabel}{#1}%
  \let\label\MYoldlabel%
}
\newcommand{\restatedef}[1]{%
  \let\MYoldlabel\label
  \let\label\MYdummylabel
  \begin{definition*}[Restatement of \cref{#1}]
    \MYload{#1}
  \end{definition*}
  \let\label\MYoldlabel
}
\newcommand{\restatetheorem}[1]{%
  \let\MYoldlabel\label
  \let\label\MYdummylabel
  \begin{theorem*}[Restatement of \cref{#1}]
    \MYload{#1}
  \end{theorem*}
  \let\label\MYoldlabel
}
\newcommand{\restatelemma}[1]{%
  \let\MYoldlabel\label
  \let\label\MYdummylabel
  \begin{lemma*}[Restatement of \cref{#1}]
    \MYload{#1}
  \end{lemma*}
  \let\label\MYoldlabel
}
\newcommand{\restateprop}[1]{%
  \let\MYoldlabel\label
  \let\label\MYdummylabel
  \begin{proposition*}[Restatement of \cref{#1}]
    \MYload{#1}
  \end{proposition*}
  \let\label\MYoldlabel
}
\newcommand{\restatefact}[1]{%
  \let\MYoldlabel\label
  \let\label\MYdummylabel
  \begin{fact*}[Restatement of \cref{#1}]
    \MYload{#1}
  \end{fact*}
  \let\label\MYoldlabel
}
\newcommand{\restate}[1]{%
  \let\MYoldlabel\label
  \let\label\MYdummylabel
  \MYload{#1}
  \let\label\MYoldlabel
}

\newcommand{\eps}{\epsilon}

\allowdisplaybreaks
\ifnum\draft=1
\newcommand{\tom}[1]{\textcolor{WildStrawberry}{[Tommaso: #1]}}
\newcommand{\francesco}[1]{\textcolor{Cerulean}{[Francesco: #1]}}
\newcommand{\marco}[1]{\textcolor{Plum}{[Marco: #1]}}
\else
\newcommand{\tom}[1]{}
\newcommand{\francesco}[1]{}
\newcommand{\marco}[1]{}
\fi

\newcommand{\om}{\om}

\newcommand{\scG}{\mathcal{G}}

\newcommand{\scA}{\mathcal{A}}

\newcommand{\AlgoName}[1]{\ensuremath{\textsc{#1}}}

\newcommand{\DenseCliqueFinder}{\AlgoName{DenseCliqueFinder}}
\newcommand{\DenseBicliqueFinder}{\AlgoName{BicliqueExtractor}}
\newcommand{\CliqueFinder}{\AlgoName{CliqueFinder}}
\newcommand{\BalancedBicliqueFinder}{\AlgoName{BalancedBicliqueFinder}}

\newcommand{\ccNP}{\mathsf{NP}}
\newcommand{\ccBPP}{\mathsf{BPP}}

\newcommand{\ignore}[1]{}

\ifpdf
\hypersetup{
    pdftitle={MY TITLE},
    pdfauthor={AUTHORS  }
}
\fi

\begin{document}
\date{}

\title{
On Finding Randomly Planted Cliques in Arbitrary Graphs
\ifnum\draft=1 {\sc \small \\ Working Draft: Please Do Not Distribute} \fi
}

\author{
Francesco Agrimonti
\and
Marco Bressan
\thanks{Università degli Studi di Milano. %\texttt{marco.bressan@unimi.it}
}
\and
Tommaso D'Orsi
\thanks{Bocconi University. %\texttt{tommaso.dorsi@unibocconi.it}}
}
}

\maketitle
\begin{abstract}
%\tom{Work in progress.} 
We study a planted clique model introduced by Feige \cite{roughgarden} where a complete graph of size $c\cdot n$ is planted uniformly at random in an arbitrary $n$-vertex graph.
We give a simple deterministic algorithm that, in almost linear time, recovers a clique of size $(c/3)^{O(1/c)} \cdot n$ as long as the original graph has maximum degree at most $(1-p)n$ for some fixed $p>0$.
The proof hinges on showing that the degrees of the final graph are correlated with the planted clique, in a way similar to (but more intricate than) the classical $G(n,\nicefrac{1}{2})+K_{\sqrt{n}}$ planted clique model.
Our algorithm suggests a separation from the worst-case model, where, assuming the Unique Games Conjecture, no polynomial algorithm can find cliques of size $\Omega(n)$ for every fixed $c>0$, even if the input graph has maximum degree $(1-p)n$.
Our techniques extend beyond the planted clique model.
For example, when the planted graph is a balanced biclique, we recover a balanced biclique of size larger than the best guarantees known for the worst case.
\end{abstract}

\ignore{
\textbf{TODO}
\begin{enumerate}\itemsep0pt
    \item controllare $\ln$ versus $\log$, usiamo entrambi, non sono sicuro sia tutto ok
    \item if $G$ has a $c$-bulging set of size $t$ then whp some prefix $v_1,\ldots,v_j$ of the vertices of $\hat G$ contains a clique of size $\frac{2j}{3}$ (resp.\ a biclique of side $\frac{j}{3}$).
    \item algorithms for dense clique/biclique
%    \item togliere vincolo $p \le \frac{1}{2}$ sostituendo $\log \frac{1}{p}$ con $\log \frac{2}{p}$
%    \item running time quasi lineare, attenzione quando si usa algo worst case
%    \item paragrafo LB
\end{enumerate}
}

\setcounter{page}{1}
\section{Introduction}\label{sec:introduction}
Finding large cliques in a graph is a notoriously hard problem. 
Its decision version was among the first problems shown to be $\ccNP$-complete \cite{karp}. In fact, it turns out that for any $\eps>0$ it is $\ccNP$-hard to find a clique of size $n^{\eps}$ even in graphs containing cliques of size $n^{1-\eps}$ \cite{hastad,zuckerman, khot_clique_inappr}.

A large body of work \cite{CLRS, halld_stillbetter,  alonkahale, kargermotwanisudan, feige, bopphalld, karakostas} focused on designing  polynomial time algorithms to find large cliques  given an $n$-vertex graph containing a clique of size $cn\,.$
When $c<1/\log n$, the best algorithm known only returns a clique of size $\Tilde{O}(\log(n)^{3})$ \cite{feige}. For larger values of $c$ it is possible to find a clique of size $O(cn)^{O(c)}$, which is of order $n^{\Omega(1)}$ when the largest clique in the graph contains a constant fraction of the vertices \cite{bopphalld, alonkahale}.
The current algorithmic landscape further suggests a phase-transition phenomenon around $c=\tfrac{1}{2}$. For sufficiently small $\eps>0$ and $c= \tfrac{1}{2}-\eps$ there exists an algorithm finding a clique of size $n^{1-O(\eps)}$ \cite{kargermotwanisudan}. Instead, when $c= \tfrac{1}{2}+\eps$, one can  efficiently find a complete graph of size $2\eps n$ via a reduction to the classical 2-approximation algorithm for vertex cover.
Finally, finding a clique of size $\Omega(\eps n)$ for $c=\tfrac{1}{2}-\eps$ was shown to be UGC-hard in \cite{khotregev,khotbansal}.

\begin{table}[H]
    \renewcommand{\arraystretch}{1.2}
    \centering
    \begin{tabular}{|l|l|l|}
    \toprule
        Regime & Output clique & References \\
    \midrule
         $c \ge \tfrac{1}{2} + \eps$ & $2 \epsilon n$ & \cite{CLRS}\\
         $c\ge \Omega(1)$ & $(cn)^{\Omega(c)}$ & \cite{alonkahale} \\
         $c \ge 1/\log n$ & $\Omega \left( \frac{n^{c}}{c} \right) $  & \cite{bopphalld, halld_stillbetter} \\
      any $c>0$ & $\Omega\left(\frac{\log^3 (cn)}{\log^2\log (cn)}\right)$ & \cite{feige}\\
        \bottomrule
    \end{tabular}
    \caption{\small Performance of state-of-the-art efficient algorithms for \maxclique\ when a clique of size $cn$ exists in the graph (note that $c$ can depend on $n$).}
    \label{tab:SOA_clique}
\end{table}

Given the grim worst-case picture, a substantial body of work has focused on designing algorithms that perform well under structural or distributional assumptions on the input graph. One research direction has investigated \maxclique and related problems on graphs satisfying expansion or colorability properties \cite{arora2011new, david2016effect, kumar2018finding, bafna2024rounding}.
Another line of work has explored  planted average case models \cite{karp,jerrum, kuvcera1995expected, alon1998finding, feige2000finding, feige2001heuristics, coja2003, feigeofek}. In the planted clique model, the input graph is generated by sampling a graph from the \Erdos-\Renyi distribution $\textnormal{ER}(n,\tfrac{1}{2})$ and then embedding a clique of size $cn$ by fully connecting a randomly chosen subset of vertices.
Here,  basic semidefinite programming relaxations \cite{feige2000finding, feige2003probable}, as well a simple rounding of the second smallest eigenvector of the Laplacian \cite{alon1998finding}, are known to efficiently recover the planted clique whenever $c\geq 1/\sqrt{n}\,.$
Lower bounds against restricted computational models further provide evidence that these algorithmic guarantees may be tight \cite{feldman2017statistical, barak2019nearly}.

\bigskip
In an effort to bridge the worst-case settings and the average-case settings, Feige and Kilian \cite{feige2001heuristics} introduced a semi-random model in which the above planted clique instance is further perturbed by: \textit{(i)} arbitrarily removing edges between the planted clique $K$ and the remainder of the graph $G\setminus K$, and \textit{(ii)} arbitrarily modifying the subgraph induced by $G\setminus K$ .  The randomness of this model lies in the cut $(K,G\setminus K)$ which separates the clique from the rest of the graph.  A flurry of works \cite{charikar2017learning, mckenzie2020new, buhai2023algorithms} led to an algorithm that, leveraging the randomness of this cut, can recover a planted clique of size $n^{\frac{1}{2}+\eps}$ in time $n^{O(1/\eps)}\,.$ This picture suggests that, from a  computational perspective, this semi-random model may be closer to the planted average case model than to worst-case graphs. (Information theoretically the semi-random model differs drastically from the planted clique model \cite{steinhardt2017does}.)

To better understand the role of randomness in the \maxclique problem, Feige  \cite{roughgarden} proposed another model in which a clique is randomly planted in an \textit{arbitrary} graph, and asked what approximation guarantees are efficiently achievable in this setting.
In comparison to the aforementioned semi-random case, here the randomness only affects  the location of the clique but not the topology of the rest of the graph. 

Investigating this model is the main focus of this paper. 
We provide a first positive answer to Feige's question, showing that a surprisingly simple deterministic algorithm achieves significantly stronger guarantees than those known for the worst-case settings, for a wide range of parameters.
Our results suggest that this model may sit in between the  average case and  the worst case regimes.

% Effort to bridge Semi-random model culminating in Rares result -> no price to pay for robustness.
% This model crucially relies on randomness of the cut.

% In an effort to better understand the role of randomness in the planted clique problem. Feige introduced the model (our) and raised the question of what approximation guarantees can be achieved for this model.

% In contrast to the aforementioned semi-random model, the source of randomness here is only the position of the clique  but not the topology of the initial graph.

% In this paper we investigate this question and show that a simple algorithm performs significantly better than existing results for the worst-case, suggesting this model may sit in between the canonical average case and worst-case settings.

\subsection{Results}\label{sec:results}
To present our contributions we first formally state our random planting model. In fact, as our results extend beyond \maxclique, the model we state is a generalization of the one in~\cite{roughgarden}.

\begin{definition}[Random planting in arbitrary graphs]\label{def:model}
Let $G$ and $H$ be graphs with $|V(H)| \le |V(G)|$. $\scG(G,H)$ describes the following distribution over graphs:
\begin{enumerate}
    \item Sample a random uniform injective mapping $\phi : V(H) \to V(G)$.
    \item Return $\hat{G}$ with $V(\hat G)=V(G)$ and  $E(\hat G) =E(G) \cup \Set{\set{\phi(u),\phi(u')}\,:\,\set{u,u'}\in E(H)}\,.$
\end{enumerate}
\end{definition}
\noindent When $H$ is the $cn$-sized complete graph $K_{cn}$, \cref{def:model} corresponds to the planted clique model of~\cite{roughgarden}.
In this specific setting we obtain the following result.
%\francesco{Non so se avete intenzione di metterlo dopo, ma credo non ci sia scritto che l'algoritmo per clique si può rendere quasilineare nella taglia del grafo per i valori di $c$ che ci interessano}
%\tom{Non c'e' specificato da nessuna parte il running time. Per fare una cosa del genere e' necessario specificarlo negli statements. Vedi zulip}

\begin{theorem}[Simplified version]\label{thm:main-clique}
There exists a deterministic algorithm $\scA$ with the following guarantees.
For every $c \in (0,1)$ and every $n$-vertex graph $G$, if $\hat{G}\sim \cG\Paren{G,K_{cn}}$ then $\scA(\hat G)$ with probability at least $1-\tfrac{1}{n^2}$ returns  a clique of size at least:
\begin{align*}
     \frac{n}{5} \cdot \Paren{\frac{c}{3}}^{\tfrac{4}{c}\log\tfrac{2}{p}}
\end{align*}
where $p = 1- \frac{\Delta}{n}$ and $\Delta$ is the maximum degree of $G$.
Moreover $\scA$ runs in time $\Tilde{O}\big(\lVert \hat G \rVert\big)$.
\ignore{
\marco{Vecchio enunciato con $1-p$ al posto di $p$.}
There exists a deterministic polynomial-time algorithm $\scA$ with the following guarantees.
For every $c \in (0,1)$, every $p \ge \nicefrac{1}{2}$, every $n$ sufficiently, and every $n$-vertex graph $G$ with maximum degree at most $pn-1$, if $\hat{G}\sim \cG\Paren{G,K_{cn}}$ then $\scA(\hat G)$ returns a clique of size at least:
\begin{align*}
    \frac{n}{5}\cdot \Paren{\frac{c}{3}}^{2+\tfrac{2}{c}\log\tfrac{1}{1-p}}\,,
    
\end{align*}
}
\end{theorem}
To appreciate the guarantees of \cref{thm:main-clique} consider the setting $p = \Omega(1)$, so that $\Delta=(1-p)n$ is bounded away from $n$.
In this case, for \emph{every} fixed $c > 0$ the algorithm of \cref{thm:main-clique} finds with high probability a clique of size $\Omega(n)$.
In the worst case, however, this is not possible unless the Unique Games Conjecture fails.
More precisely, assuming UGC, no polynomial-time algorithm can find a clique of size $\eps\cdot n$ even when one of size $\Paren{\tfrac{1}{2}-\eps} \cdot n$ exists~\cite{khotregev,khotbansal}; indeed, state-of-the-art algorithms \cite{alonkahale,bopphalld, halld_stillbetter} are only known to return cliques of size $n^{O(c)}$.
By adding $n\frac{p}{1-p}$ isolated vertices to the graph, it also follows that under UGC one cannot efficiently find a clique of size $\eps \cdot n$ even when one of size $\Paren{\tfrac{1-p}{2}-\eps} \cdot n$ exists \emph{and} the input graph has degree $\Delta \le (1-p)n$, as in the statement of \cref{thm:main-clique}.
Thus, unless UGC fails, we cannot expect \cref{thm:main-clique} to hold in the worst case.
We remark that \cref{thm:main-clique} also guarantees to recover cliques of size $n^{\Omega(1)}$ for $c\geq \Omega\Paren{\tfrac{\log\log n}{\log n}}$, a regime in which worst-case algorithms are only known to find cliques of size $\poly\log(n)$.

Note that the performance of our algorithm deteriorates as $p$ approaches $0$; that is, as the maximum degree approaches $n$. While it remains an open question whether some assumption on the degree is inherently necessary, we provide some preliminary evidence in \cref{thm:densification_LB}, see \cref{sec:techniques} and \cref{sec:lower-bound}.
%We also remark that the constraint $p \le \nicefrac{1}{2}$ is somewhat arbitrary, and serves the only purpose of keeping the bound in the range reached by the algorithm.
Finally, as one can expect, the failure probability can be actually made smaller than $n^{-a}$ for any desired $a \ge 1$; see the full formal version of \cref{thm:main-clique} in \cref{sec:clique}.
% for 1-p>Omega(1) linear size for constant c and polynomial for c>=loglog n/log n. Then rapidly match worst-case algortihms

% State of the art is *polynomial* for constant c and polylogarithmic for c = Theta(loglog n/log n)

% Linear for c>1/2 + eps

% Beyond UGC-hardness

\bigskip
Our results extend beyond the case where the planted graph is a complete graph. To illustrate this, we also consider the \maxbalancedbiclique problem, where the goal is to find a largest complete balanced bipartite subgraph.
The \maxbalancedbiclique problem has a long history \cite{gareyjohnson, johnson_column, alon_algo_aspects_regularity_lemma} and a strong connection to \maxclique \cite{chalermsook}.
Assuming the Small Set Expansion Hypothesis, there is no polynomial-time algorithm that can find a balanced biclique within a factor $n^{1-\eps}$ of
the optimum for every $\eps>0$, unless $\ccNP\subseteq \ccBPP$ \cite{manurangsi_complessita_biclique}. Remarkably, in the worst case, the bicliques that existing algorithms are known to return are  significantly smaller than the complete graphs found in the context of \maxclique.
In fact, the best algorithm known \cite{chalermsook} works through a reduction to \maxclique which constructs an instance with a complete graph of size $O(c^2\cdot n)$ from a \maxbalancedbiclique instance with a biclique of size $c\cdot n\,.$

In comparison, under \cref{def:model}, we obtain the following guarantees.

%\francesco{Nel seguente enunciato ho cambiato la probabilità $1-\tfrac{1}{n}$ in $1-\tfrac{1}{n^{2}}$ (simmetria con clique)}
\begin{theorem}[Simplified version]\label{thm:main-biclique}
There exists a deterministic polynomial-time algorithm $\scA$ with the following guarantees.
For every $c \in (0,1)$ and every $n$-vertex graph $G$,
if $\hat{G}\sim \cG\Paren{G,K_{\frac{cn}{2}, \frac{cn}{2}}}$ then $\scA(\hat G)$ with probability at least $1-\tfrac{1}{n^{2}}$ returns a balanced biclique of size at least:
\begin{align*}
    \frac{c}{48} \cdot 2^{ \sqrt{\frac{c \log n}{2}} }.
\end{align*}
Moreover $\scA(\hat G)$ runs in time $\Tilde{O}(\lVert \hat G \rVert)$.% for every input graph $\hat G$.
%\marco{Le costanti nel bound sono quelle vecchie di quando l'algoritmo non era ottimizzato. Con il running time $\Tilde{O}(\lVert \hat G \rVert)$ bisogna sistemarle, tipo mettendo $\frac{c}{4}$ invece di $\frac{c}{3}$.} \francesco{Dovrei aver sistemato, è $\frac{c}{12}$ al posto di $\frac{c}{3}$. Inoltre: in "returns a balanced biclique of size at least" non c'era "at least", l'ho aggiunto.}
\end{theorem}
The main point of \cref{thm:main-biclique} is again the difference with the worst case bounds.
In the worst case, existing algorithms are known to find a biclique of size $(\log n)^{\omega(1)}$ only if there exists one of size $c\cdot n\geq \omega\Paren{\tfrac{\log\log n}{\sqrt{\log n}}}\cdot n$ in the input graph.
In contrast, \cref{thm:main-biclique} states that in typical instances from  $\cG\Paren{G,K_{\frac{cn}{2}, \frac{cn}{2}}}$ we can efficiently find a biclique of size $(\log n)^{\omega(1)}$ whenever there exists one of size $c\cdot n\geq \omega\Paren{\tfrac{\log^2\log n}{\log n}}\cdot n$; that is, for value of $c$ up to $\frac{\log\log n}{\sqrt{\log n}}$ times smaller than for the worst case.
Furthermore, unlike the bounds of \cref{thm:main-clique}, the ones of \cref{thm:main-biclique} are insensitive to the structure of $G$, and in particular to its maximum degree.

\section{Techniques}\label{sec:techniques}
This section gives an intuitive description of our techniques, using the planted clique problem as a running example.
Let $G$ be an arbitrary $n$-vertex graph, and let $\hat{G}\sim \cG(G,K_{cn})$.
For simplicity, we suppose that $c>0$ and $p>0$ are fixed constants, and that $G$ has maximum degree $\Delta \le (1-p)n$. Let $\phi\,:V(K_{cn})\to V(G)$ be the injective mapping sampled in the process of constructing $\hat{G}\,.$ 
Because we have almost no knowledge of the global structure of $G\,,$ it appears difficult to recover the planted clique via the topology of $\hat{G}$  without running into any of the barriers observed in worst-case instances.

On the other hand, since the clique is planted randomly, we can expect certain basic statistics to change in a convenient and somewhat predictable way between 
$G$ and $\hat G\,.$
Our approach focuses on perhaps the simplest such statistic---the degree profile---guided by the intuition that vertices with higher degree in $\hat G$ are more likely to belong to the planted clique than those with lower degree. For notational convenience we use the degree in the complement graph, which we call \emph{slack}.
To be precise, for a vertex $v \in V(G)$, the slack of $v$ in $G$ is $s_v=(n-1)-d_v$ where $d_v$ is the degree of $v$ in $G$.
In the same way we define the slack of $v$ in $\hat G$ as $\hat s_v=(n-1)- \hat d_v$, where $\hat d_v$ is the degree of $v$ in $\hat G$.

To formalize the intuition above, suppose  $G$ contains a subset $V'\subseteq V$ such that \textit{(i)} the vertices of $V'$ have approximately the same slack, in the sense that if $s:=\min_{v\in V'} s_v$, then any $v\in V'$ satisfies
\begin{align*}
    s_v \Paren{1-\tfrac{c}{2}} < s\,,
\end{align*}
and \textit{(ii)} the set $V_{< s}(G)$ of vertices in $G$ with slack smaller than $s$ has size at most, say, $\tfrac{c}{10}\cdot\card{V'}\,.$
Because the map $\phi$ is chosen uniformly at random, we expect a $c$ fraction of $V'$ will be in the image of $\phi\,.$ Furthermore, every $v\in V'$ in the image of $\phi$ acquires $c s_v$ new neighbors in expectation, which by \textit{(i)} gives:
\begin{align*}
    \E_\phi \Brac{\hat s_v}\leq s_v \cdot (1 - c) < s\,.
\end{align*}
In fact, as long as $\card{V'}$ and $s$ are large enough (roughly $\Omega(c^-1 \log n)$), by standard concentration bounds at least $\frac{c}{2} |V'|$ vertices of $V'$ will be in the image of $\phi$, and all those vertices $v$ will satisfy $\hat s_v < s$.
Under these circumstances, by \textit{(ii)} we conclude that, in $\hat G$, the set $\card{V_{< s}(\hat G)}$ of vertices having slack smaller than $s$ has size at least $\frac{c}{10}|V'|$, and moreover a fraction at least $\frac{\nicefrac{1}{2}}{\nicefrac{1}{2}+\nicefrac{1}{10}} > 0.8$ of those vertices form a clique.
%\begin{align*}
%    \card{V_{< s}(\hat G)}\approx \card{V_{<s}(G)} + c\card{V'}\leq \tfrac{c}{3}\card{V'} + c\card{V'}\,,
%\end{align*}
%while also more than half of $\card{V_{< s}(\hat G)}$ lives in the image of $\phi\,.$
We can then immediately recover a clique of size $\Omega(c\card{V'})$ via the standard reduction to vertex cover applied to the subgraph of $\hat{G}$ induced by $\card{V_{< s}(\hat G)}\,.$

The above discussion suggests our intuition is correct whenever a sufficiently large set satisfying \textit{(i)} and \textit{(ii)} exists.\footnote{We remark that our algorithm does not need to find this set.}
While arbitrary graphs may not contain such a set, it turns out that the only obstacle towards the existence of a linear size set $V'$ is the presence of a large set of vertices of slack strictly smaller than $s$. Choosing $V'$ so that $s \le p\cdot n$ we deduce that such a $V'$ must exists.

\begin{remark}
    The above reasoning works beyond the parameters regime of our example and, in fact, does not require the planted graph to be a clique. In the context of \maxbalancedbiclique the existence of a large set with slack $<s$ makes the problem easier. Therefore, we are able to drop the assumption on the maximum degree in $G\,.$
\end{remark}

We complement the intuition above with a lower bound on the performance of degree profiling.
Essentially this states that, if we have no guarantees on the maximum degree of $G$, then the degree profile of $\hat G \sim \scG(G,K_{cn})$ is uncorrelated with $K_{cn}$.
Formally:
\begin{theorem}\label{thm:densification_LB}
For every $c \in (0,\frac{1}{2})$ and $n \ge 3$ there exists an $n$-vertex graph $G$ such that $\hat G \sim \scG(G,K_{cn})$ satisfies what follows with probability at least $1-\frac{1}{n}$.
For every ordering $v_1,\ldots,v_n$ of the vertices of $\hat G$ by nonincreasing degree, and for every $j\in [n]$, the largest clique in the induced subgraph $\hat G\brac{\set{v_1,\ldots,v_j}}$ has size at most 
\begin{align}\label{eq:LB_clique_bound}
    O\left( \frac{\sqrt{n \ln n}}{c} + c \, j\right)\,.
\end{align}
\end{theorem}
\noindent 
To appreciate \cref{thm:densification_LB} let $t = \Omega(c^{-2} \sqrt{n \ln n})$.
Then the theorem says that, if one takes the first $t$ vertices of $\hat G$ in order of degree, the largest clique therein has size $O(ct)$ with high probability.
In other words, for all $t$ not too small compared to $n$, the $t$ vertices of highest degree have roughly the same clique density of the entire graph.
This suggests that, using degree statistics alone, one has little hope to find cliques larger than $\tilde{O}(\sqrt{n})$ even for constant $c$.
Note that there is no contradiction with the upper bounds of \cref{thm:main-clique}: those bounds become trivial for large $\Delta$, and the graph behind the proof of \cref{thm:densification_LB} has indeed a large $\Delta$.
%Note that the theorem does not bound the maximum degree of $G$, and indeed the graph $G$ in the construction has maximum degree close to $n$.
%This proves that the bound of \cref{thm:main-clique} is tight for \cref{alg:mainclique}, at least qualitatively, in its dependence on $p$.
%Said otherwise, degree profiling alone cannot find cliques of large size (say, $\Omega(n^{0.51})$) for small values of $p$.

% Pick simplest value of c.
% No leverage on the structure of G -> unclear how to explore the global structure
% TO circumvent this issue we focus on the simplest statistics that may significantly differ between G and \hat{G}. Namely, the degree profiles of the vertices in the graph.

%\tom{@Marco could you add a paragraph about the lower bound?}
\section{Preliminaries}\label{sec:preliminaries}

Let $G$ be a graph.
We let $V(G)$ be its set of vertices and $E(G)$ its set of edges. 
For $V'\subseteq V(G)$ we let $G[V']$ be the subgraph induced by $V'\,.$
We often use $n=\card{V(G)}$.
We let $\lVert G \rVert = |V(G)|+|E(G)|$.
For $v\in V(G)\,,$ let $d_v$ be its degree and $s_v:=n-1-d_v$ its slack. For $V'\subseteq V(G)\,,$ let $s_{V'}:=\min_{v\in V'}s_v\,.$
We write $V_{<s}(G):=\set{v\in V(G)\,|\, s_v < s}\,.$ 
We do not specify the graph when the context is clear and we define $V_{< s} (\hat{G})$ as $\hat{V}_{<s}$.
We let $K_{n}$ be the complete graph of size $n$ and $K_{a,b}$ be the biclique with sides of size $a$ and $b$.
We let $[n]:=\set{1,\ldots, n}\,, $ $\log=\log_2$ and $\ln=\log_e$. 

The computational model is the standard RAM model with words of logarithmic size.
Unless otherwise stated, all our graphs are given as adjacency list.
By performing a $O(n)$ preprocessing we henceforth assume the adjacency lists are sorted, so that one can perform binary search and check the existence of any given edge in time $O(\log n)$.

\bigskip
The following theorem says that, for every fixed $c > \frac{1}{2}$, one can efficiently find a clique of size $\Omega (n)$ in an $n$-vertex graph that contains one of size $cn$.
\begin{theorem}\label{thm:dense_clique_finder}
    There exists an algorithm, \textnormal{\DenseCliqueFinder}, with the following guarantees. For every $\varepsilon > 0$, if $\scA$ is given in input an $n$-vertex graph $G = (V, E)$ that contains a clique of size $(\frac{1}{2} + \varepsilon)n$, then $\scA$ finds in deterministic $\tilde{O}(n^2)$-time a clique of size $2 \varepsilon n$.
\end{theorem}
\noindent The proof is folklore---take the complement of $G$, find a $2$-approximation of the smallest vertex cover through a maximal matching, and return its complement. See also \cite{CLRS}.

%As our guarantees are trivial for $c\leq O(1/\log n)$, throughout the rest of the work we tacitly assume $c \in \omega\left(\frac{1}{\log n}\right)\,.$ \marco{Check this!}
\section{Slackness profile and densification}\label{sec:slackness}

In this section we prove our structural results on the degree and slackness profile of graphs from \cref{def:model}.
%Let $G=(V,E)$ be a graph and $n=|V|$. For every $s=1,\ldots,n$ let $V_{< s} = \{v \in V: s_v < s\}$ and $n_{< s} = |V_{< s}|$.
We start with a definition.
\begin{definition}[Bulging set]\label{def:good_set}
%\marco{$c$-outstanding? $c$-carved? Nel senso che si ``distingue'' con forza $c$}
Let $G=(V,E)$ be a graph and $\alpha,\beta > 0$.
A set $U \subseteq V$ is \emph{$(\alpha, \beta)$-bulging} if:
\begin{enumerate}\itemsep0pt
    \item $s_v < \frac{s_U}{1-\beta}$ for all $v \in U\,.$
    \item $|V_{< s_{U}}| < \frac{1}{\alpha} |U|\,.$
\end{enumerate}
%where $s_U = \min\{s_v:v \in U\}$.
\end{definition}

The next statement characterizes the existence of $(\alpha, \beta)$-bulging sets in \textit{any} graph based on the value of $\alpha$ and $\beta$ and the slackness  of its vertices.
%\francesco{Non capisco il "The first statement" sopra: si intende "The next"?}
\begin{lemma}\label{lem:structural}
Let $G=(V,E)$ be an $n$-vertex graph.
Then for every $\beta \in (0,\nicefrac{1}{2})$, $\alpha \ge 2$, and $s \in \R_{>0}$ at least one of the following facts holds:
\begin{enumerate}[(i)]
    \item $|V_{<s}| \ge \frac{n}{ \alpha ^{2 + \frac{1}{\beta}\log \frac{n}{s}}}$.
    \item $G$ contains an $(\alpha, \beta)$-bulging set $U$ such that $|U| \ge \frac{n}{\alpha^{1 + \frac{1}{\beta}\log \frac{n}{s}}}$ and $s_U \ge s$.
\end{enumerate}
\end{lemma}
%\francesco{è necessario usare/definire $n_{<s}$? Forse usare $|V_{<s}|$ eviterebbe di introdurre un'altra definizione. Ditemi voi}\tom{meglio usare $|V_{<s}|$ per me}
\begin{proof}
Let  $\eta = \frac{\beta}{1-\beta} > 0$ and $h = \left\lceil \log_{1+\eta} \frac{n}{s} \right\rceil$.
We define a partition of $V$ into $h+1$ possibly empty sets, as follows:
\begin{align}
    V_0 &:= V_{<s} \\
    V_j &:= V_{<s(1+\eta)^j} \setminus V_{<s(1+\eta)^{j-1}} = \left\{ v \in V \, \middle| \, s(1 + \eta)^{j - 1} \le s_v < s(1 + \eta)^j \right\} \quad j \in [h]
\end{align}
It is immediate to see that this is indeed a partition of $V$, since $0 \le s_v < n$ for every $v \in V$. 
We prove the statement by contradiction.
Suppose \textit{(ii)} does not hold.
Then it must be that for $j\geq 1$:
\begin{equation}\label{eq:main_lemma_contradiction}
    |V_j| < \frac{n}{\alpha^{2 + \frac{1}{\beta}\log \frac{n}{s}}} \cdot \alpha^j\,,
\end{equation}
Indeed, if this was not the case, then one can check that for the smallest $j\in [h]$ violating \cref{eq:main_lemma_contradiction} the set $V_j$ would be $(\alpha, \beta)$-bulging, and moreover every vertex in $V_j$ would have slack at least $s$ (since $j \ge 1$).
Suppose further \textit{(i)} is not verified. Then as the sets $V_j$ form a partition of $V$, and as $\alpha \ge 2$,
\begin{align}
    n = \sum_{j=0}^h |V_j| < \frac{n}{\alpha^{2 + \frac{1}{\beta}\log \frac{n}{s}}} \cdot \sum_{j=0}^h \alpha^j < \frac{n}{\alpha^{2 + \frac{1}{\beta}\log \frac{n}{s}}} \cdot \alpha^{h+1}
    \label{eq:n<n}
\end{align}
Now observe that
\begin{align}
    h \le 1 + \frac{\log\frac{n}{s}}{\log(1+\eta) } \le 1 + \frac{1+\eta}{\eta}\log\frac{n}{s} = 1 + \frac{1}{\beta} \log \frac{n}{s}
\end{align}
where we used the facts that $\log (1+x) \ge \frac{x}{1+x}$ for all $x \ge 0$, and that $\frac{\eta}{1+\eta}=\beta$.
Substituting this bound in \Cref{eq:n<n} yields the absurd $n<n$.
Thus at least one among \textit{(i)} and \textit{(ii)} holds.
\end{proof}

Our next key result states that the subgraph of $\hat G$ induced by the set of vertices $v$ with slack $\hat{s}_{v} < s_{U}$, where $U \subseteq V$ is the bulging set that exists in $G$ for \Cref{lem:structural}, will contain a large number of vertices of $H$ with high probability.

\begin{lemma}[Densification Lemma]\label{lem:good_set_densifies}
Let $G$ be an $n$-vertex graph, $\alpha \ge 2$ and $c\in(0,1)\,.$ Let $H$ be a regular graph with $|V(H)| \le n$ and minimum degree at least $cn \ge 10$. Let $U$ be an $(\alpha, \frac{c}{2})$-bulging set of $G$ with $\min\{s_U, |U|\} \ge \frac{12 + 29 a \ln n}{c}$ for some $a \ge 1$.
Finally, let $\hat G \sim \scG(G,H)$, and let $\hat H$ be the image of $H$ in $\hat G$.
Then, with probability at least $1-n^{-a}$ the set $\hat V_{< s_U}$ satisfies:
\begin{enumerate}[(i)]
    \item $|\hat V_{< s_U} \cap H| > \frac{c}{2} \cdot |U|$.
    \item $|\hat{ V}_{< s_U} \cap H| >  \frac{c \alpha}{2} \cdot |\hat{V}_{< s_U} \setminus H|$.
\end{enumerate}
\end{lemma}
\begin{proof}
    For brevity, let $S:=\hat V_{< s_U}\,.$
    First, we claim that $H \cap U \subseteq \hat V_{< s_U}$ with high probability.
    Consider any $v \in U$, and note that $v \notin \hat V_{< s_U}$ means $\hat s_v \ge s_U$.
    Now, if $v \in H$, then $\hat s_v = s_v - X$, where $X=\sum_{i=1}^{s_v} X_i$ is the sum of non-positively correlated Bernoulli random variables of parameter $c' = c-\frac{1}{n}$.
    The event $\hat s_v \ge s_U$ is therefore the event $X \le s_v-s_U$; since $s_v-s_U \le \frac{c}{2}s_v$, as $v \in U$ and $U$ is $(\alpha, \frac{c}{2})$-bulging, this implies the event $X \le \frac{c}{2}s_v$.
    Now, as $cn \ge 10$, then $c' \ge \frac{9}{10}c$, and $\frac{c}{2}s_v \le \frac{5}{9} c' s_v$.
    Since moreover $\E X = c' s_v$, we conclude that $\hat s_v \ge s_U$ implies the event $X \le (1-\nicefrac{4}{9}) \E X$.
    We then use \Cref{lem:chernoff} with $\varepsilon = \nicefrac{4}{9}$.
    To this end note that:
    \begin{align}
        \E X = c'\, s_v \ge \frac{9}{10} c\, s_U > 10 + 26 a \ln n \ge 13 \left(\ln 2 + (1+a)\ln n\right) = 13 \ln \left(2 n^{a+1}\right)
    \end{align}
    Therefore:
    \begin{align}
        \Pr[v \notin \hat V_{< s_U}] = \Pr[\hat s_v \ge s_U] \le \Pr[X \le (1-\nicefrac{4}{9}) \E X] \le e^{-\frac{(\nicefrac{4}{9})^2}{2+\nicefrac{4}{9}} \E X} < e^{-\frac{\E X}{13}} < \frac{1}{2 n^{a+1}}
    \end{align}
    By a union bound over all $v \in U$ we conclude that $H \cap U \not\subseteq \hat V_{< s_U}$ with probability at most $\frac{1}{2} n^{-a}$.
    
    Next, consider $|H \cap U|$. Note that $|H \cap U| = X$ where again $X=\sum_{i=1}^{s_v} X_i$ is the sum of non-positively correlated Bernoulli random variables of parameter $c-\frac{1}{n}$.
    Using again \Cref{lem:chernoff} with $\varepsilon = \nicefrac{4}{9}$, and noting as done above that $\E X = c' |U| > 13 \ln \left(2 n^{a+1}\right)$, we obtain:
    \begin{align}
        \Pr\left[|H \cap U| \le \frac{c}{2}|U| \right] = \Pr\!\left[X \le \big(1-\nicefrac{4}{9}\big)\E X\right] < \frac{1}{2 n^{a+1}} < \frac{1}{2}n^{-a}
    \end{align}

    Finally, let $S=\hat V_{< s_U}$.
    The bounds above show that, with probability at least $1 - n^{-a}$, we have $U \cap H \subseteq S \cap H$ and $|U \cap H| > \frac{c}{2}|U|$, which implies $|S \cap H| > \frac{c}{2}|U|$, that is, (i).
%    Since moreover $|S| \ge |H \cap U| \ge \frac{c}{2}|U|$, proving \textit{(i)}.
    Moreover $S \setminus H \subseteq V_{< s_U}$, which implies $|S \setminus H| < \frac{1}{\alpha}|U|$ as $U$ is $(\alpha, \frac{c}{2})$-bulging. Together with (i) we conclude that:
    \begin{align}
        \frac{|S \cap H|}{|S \setminus H|} >  \frac{c \alpha}{2}\,
    \end{align}
    which proves \textit{(ii)}.
\end{proof}

\section{Application to \maxclique}\label{sec:clique}

In this section we prove \cref{thm:main-clique}, which we restate in a fully formal way and with more general probabilistic guarantees.

\begin{theorem}\label{thm:main-clique-formal}
There exists a deterministic algorithm $\scA$ with the following guarantees.
Fix any $a \ge 1$.
Let $c := c(n) \in \omega\left(\tfrac{1}{\log n}\right)$, and define:
\begin{align*}
    K(n,c,p) \coloneq \frac{n}{5}\cdot \Paren{\frac{c}{3}}^{2+\tfrac{2}{c}\log\tfrac{2}{p}}\,.
\end{align*}
Then for every $n$ large enough and every $n$-vertex graph $G$ what follows holds.
Letting $p=1-\frac{\Delta}{n}$ where $\Delta$ is the maximum degree of $G$, if $K(n,c,p) \ge 1 + 2a \ln n$, then $\scA$ on input $\hat{G}\sim \cG\Paren{G,K_{cn}}$ returns a clique of $\hat G$ whose size is at least $K(n,c,p)$ with probability at least $1-n^{-a}$.
Moreover $\scA(\hat G)$ runs in time $\Tilde{O}(\lVert \hat G \rVert)$ for every input graph $\hat G$.
\end{theorem}

\begin{proof}
We start by proving that \cref{alg:mainclique} runs in time $\tilde{O}(n^3)$ and guarantees a clique of size $\frac{7}{5} K(n,c,p)$ with the prescribed probability.
We then show how to lower the running time to $\tilde{O}(\lVert \hat G \rVert)$ while reducing the clique size to $K(n,c,p)$.

\begin{algorithm}[H]
\caption{\CliqueFinder($\hat G$)}\label{alg:mainclique}
\begin{algorithmic}[1]
\State $S \leftarrow \emptyset$
\State $v_1, \ldots , v_n \leftarrow$ vertices of $\hat G$ in nonincreasing order of degree
\For{$1 \le i \le n$:} \label{alg:mainclique:line:for}
\State $T \leftarrow \DenseCliqueFinder \left( \hat G \left[ \{ v_1,\ldots,v_i \} \right]  \right)  $  
\If{ $\left| T \right| > \left| S \right| $:}
\State $S \leftarrow T$
\EndIf
\EndFor
\State \Return $S$
\end{algorithmic}
\end{algorithm}

%In fact, for \cref{alg:mainclique} we prove a slightly more general claim; namely, the probability of success is at least $1-n^{-a}$ for any fixed $a \ge 1$, as long as $K(n,c,p) \ge 1 + 2 a \ln n$ and $n$ is sufficiently large.
%Substituting $a=2$ (and lowering the running time as said above) yields the bounds of \cref{alg:mainclique}.
The inequalities we are going to claim assume $n$ is indeed sufficiently large (formally, larger than some $n_0$ that may depend on $a$).
To begin with, observe that if $c \le \frac{1}{\log n}$ or $p \le n^{-\nicefrac{1}{2}}$ then $K(n,c,p) \le 1$ and therefore our algorithm certainly satisfies the bound of \cref{thm:main-clique-formal}.
Indeed, if $c \le \frac{1}{\log n}$ then the second multiplicative term in the expression of $K(n,c,p)$ satisfies:
    \begin{align}
        \Paren{\frac{c}{3}}^{2+\tfrac{2}{c}\log\tfrac{2}{p}} \le \Paren{\frac{1}{3 \log n}}^{2+2 \log n} < 9^{-\log n} \le \frac{5}{n}
    \end{align}
If instead $p \le n^{-\nicefrac{1}{2}}$ then the same term satisfies:
\begin{align}
    \Paren{\frac{c}{3}}^{2+\tfrac{2}{c}\log\tfrac{2}{p}} \le \Paren{\frac{1}{3}}^{2 \log \sqrt{n}} = 3^{-\log n} \le \frac{5}{n}
\end{align}
Thus we may assume $cp \ge \frac{1}{\sqrt{n} \log n}$, and therefore:
\begin{align}
    cp \ge \frac{13 + 29 a \ln n}{n} \ge \frac{12 + 29 a \ln n}{n} + \frac{c}{n} \label{eq:cp_LB}
\end{align}

Now let $s=pn-1$.
Then $s=(n-1)-\Delta$; hence all vertices of $G$ have slack at least $s$, and therefore $|V_{<s}| = 0$.
Now apply \Cref{lem:structural} with $\alpha = \frac{3}{c}$ and $\beta = \frac{c}{2}$.
Note that item \textit{(i)} fails, thus item \textit{(ii)} holds.
Therefore $G$ contains a $(\frac{3}{c}, \frac{c}{2})$-bulging set $U$ such that:
\begin{align}
|U| \ge \frac{n}{\left(\frac{3}{c}\right)^{1 + \frac{2}{c}\log \frac{n}{s}}}
= n \cdot \left(\frac{c}{3}\right)^{1 + \frac{2}{c}\log \frac{n}{s}}
\ge n \cdot \left(\frac{c}{3}\right)^{1 + \frac{2}{c}\log \frac{2}{p}}
\ge \frac{15}{c} \cdot K(n,c,p) 
\end{align}
where we used the fact that $p \ge \frac{1}{\sqrt{n}}$ and that $n$ is large enough to obtain that $\frac{n}{s} \le \frac{2}{p}$.
Thus, when $K(n,c,p) \ge 1 + 2 a \ln n$ we have $|U| \ge \frac{12 + 29 a \ln n}{c}$.
%This the bound of the claim, and thus in $\omega(c^{-1} \ln n)$.
%In particular, for every $a \ge 1$, for every sufficiently large $n$ we have $|U| \ge \frac{12 + 29 a \ln n}{c}$.
Moreover $s_U \ge s = p n$; using \cref{eq:cp_LB} this yields $s_U \ge \frac{12 + 29 a \ln n}{c}$, too.
We can then apply \Cref{lem:good_set_densifies}.
It follows that, with probability at least $1-n^{-a}$, we have $\frac{|\hat{ V}_{< s_U} \cap H|}{|\hat{V}_{< s_U} \setminus H|} >  \frac{c \alpha}{2} > \frac{3}{2}$.
We deduce that $G[\hat{V}_{<s_{U}}]$ contains a clique whose density is at least:
\begin{align}
    \frac{|\hat V_{<s_{U}} \cap H|}{|\hat V_{<s_{U}}|} = \frac{|\hat V_{<s_{U}} \cap H|}{|\hat V_{<s_{U}} \cap H| + |\hat V_{<s_{U}} \setminus H|} >\frac{|\hat V_{<s_{U}} \cap H|}{ (\frac{2}{3} + 1) |\hat V_{<s_{U}} \cap H|} = \frac{1}{2} + \frac{1}{10}.
\end{align}
With the same probability we have simultaneously that $|\hat V_{< s_U}| \ge \frac{c}{2}|U|  > 7 \cdot K(n,c,p)$.
Now consider the invocation of \DenseCliqueFinder\ on $G[\hat{V}_{<s_{U}}]$.
By \Cref{thm:dense_clique_finder}, that invocation finds a clique of size at least:
\begin{align}
    7 \cdot K(n,c,p) \cdot \left(2 \cdot \frac{1}{10}\right) = \frac{7}{5} \cdot K(n,c,p)
\end{align}

Next, we bring the running time in $\tilde{O}(n^2)$ while ensuring an output clique of size $K(n,c,p)$.
To this end, change the loop at line \ref{alg:mainclique:line:for} so as to iterate only over $i$ in the form $i=(1+\eta)^j$ for some $\eta > 0$.
For the smallest $i=(1+\eta)^j$ such that $V_{<s_{U}} \subseteq \{v_1,\ldots,v_i\}$ the subgraph $\hat G[{v_1,\ldots,v_i}]$ will then have clique density $\frac{\frac{1}{2}+\frac{1}{10}}{1+\eta}$ and will contain $\hat V_{< s_U}$ plus at most $\eta |\hat V_{< s_U}|$ other vertices.
By choosing $\eta>0$ sufficiently small one can then ensure that \DenseCliqueFinder\ when ran on $\hat G[{v_1,\ldots,v_i}]$ returns a clique of size at least $K(n,c,p)$.
The total number of iterations is obviously in $O(\log_{1+\eta} n)=O(\log n)$, and by \cref{thm:dense_clique_finder} every iteration takes time $\tilde{O}(n^2)$, giving a total time of $\tilde O(n^2)$ too.

To finally bring the running time in $\tilde{O}(\lVert \hat G \rVert)$, upon receiving $\hat G$ we check whether $\lVert \hat G \rVert \le \binom{n / \log n}{2}$.
If that is the case then $c \le \frac{1}{\log n}$ and $K(n,c,p) \le 1$ as shown above; in this case we return any vertex of $\hat G$.
Otherwise, we run the algorithm above.
In both cases the bounds are satisfied and the running time is in $\tilde{O}(\lVert \hat G \rVert)$.
\end{proof}

%\francesco{Idee Prof. Bressan:} 

%\francesco{1) Sarebbe il caso di dire che la probabilità di fallire può essere resa 1/poly di grado alto a piacere}

%\francesco{2) mettere i teoremi con poly}

\section{Application to \maxbalancedbiclique}\label{sec:biclique}
In this section we restate and prove a more formal and general version of \cref{thm:main-biclique}: 
\begin{theorem}\label{thm:main-biclique-formal}
There exists a deterministic polynomial-time algorithm $\scA$ with the following guarantees. 
Fix any $a \ge 1$. Let $c := c(n) \in \omega\left(\tfrac{1}{\log n}\right)$. For every $n$ large enough and every $n$-vertex graph $G$,
when given $\hat{G}\sim \cG\Paren{G,K_{\frac{cn}{2}, \frac{cn}{2}}}$ in input, $\scA$ with probability at least $1-n^{-a}$ returns a balanced biclique of size at least:
\begin{align*}
    \frac{c}{48} \cdot 2^{ \sqrt{\frac{c \log n}{2}} }.
\end{align*}
Moreover $\scA(\hat G)$ runs in time $\Tilde{O}(\lVert \hat G \rVert)$ for every input graph $\hat G$.
%\marco{Le costanti nel bound sono quelle vecchie di quando l'algoritmo non era ottimizzato. Con il running time $\Tilde{O}(\lVert \hat G \rVert)$ bisogna sistemarle, tipo mettendo $\frac{c}{4}$ invece di $\frac{c}{3}$.} \francesco{Dovrei aver sistemato, è $\frac{c}{12}$ al posto di $\frac{c}{3}$. Inoltre: in "returns a balanced biclique of size at least" non c'era "at least", l'ho aggiunto.}

\end{theorem}
\francesco{ho inserito "thm:main-biclique-formal" nella label ed ho già sostituito i link alla versione informale con quella formale}

The algorithm behind the theorem, \cref{alg:mainbiclique}, is based on the following intuition.
Observe that our main technical result, \Cref{lem:structural}, essentially says that every graph $G$ contains either (i) a large number of vertices of small slack (and thus large degree), or (ii) a large bulging set.
If (i) holds, then we can hope to find a large biclique by just intersecting the neighborhoods of those vertices (namely, of the $k$ vertices with largest degree for some suitable value of $k$).
If (ii) holds, then we can hope to find a large biclique by exploiting the ``densification'' phenomenon used by our clique algorithm (see \cref{sec:clique}).
The structure of the algorithm follows this intuition, with a first phase that finds a large biclique if (i) holds and a second phase that finds a large biclique if (ii) holds.
\ignore{
The algorithm we use sorts the vertices in non-increasing order by degree, next it can be split in two phases, that capitalize on the two possible items in the statement of \Cref{lem:structural}, one of which must certainly hold. 
In its first phase, if the first item holds the algorithm finds a balanced biclique of suitable size: in fact, in this case $\hat{G}$ has many vertices of high degree, so a balanced biclique can be found easily by intersecting their neighborhoods. If the first item doesn't hold, then the second item of \Cref{lem:structural} does, and as a consequence also \Cref{lem:good_set_densifies} does: this means that with high probability in $\hat{G}$ there exists a subset of vertices contiguous in the ordering in which the vertices of $H$ are dense. The second phase of the algorithm is designed in order to find this dense subset, and exploit its density by applying to it the algorithm of \Cref{lem:bicliquepoly}. }
\begin{algorithm}[H]
\caption{$\BalancedBicliqueFinder (\hat{G})$}\label{alg:mainbiclique}
\begin{algorithmic}[1]
\State $v_1, \ldots , v_n \leftarrow$ vertices of $\hat{G}$ in non-increasing order of degree
\State $L, R \leftarrow \emptyset$
%\State $t \leftarrow$ see~\Cref{eq:def_t}
\For{$1 \le i \le n$:} \Comment{Phase 1} \label{algo:biclique:line:P1}
\State $L' \leftarrow \left\{ v_1, \ldots, v_{i }\right\}$ 
\State $R' \leftarrow \bigcap_{v \in L'} N_v$ 
%\State $R \leftarrow \text{ any $|L|$ vertices of } \bigcap_{v \in L_1} N_v$
%\State $L_2, R_2, L', R' \leftarrow \emptyset$ \Comment{Phase 2}
\If{$\min \{|L'|, |R'| \} > \min \{|L|, |R| \}$}
\State $L, R \leftarrow L', R'$
\EndIf
\EndFor
\For{$1 \le i < j \le n$:} \Comment{Phase 2} \label{algo:biclique:line:P2}
\State $L', R' \leftarrow \DenseBicliqueFinder(\hat{G}, \{ v_i,\ldots,v_j \})$ \label{algo:biclique:line:dense}
\If{$\min \{|L'|, |R'| \} > \min \{|L|, |R| \}$}
\State $L, R \leftarrow L', R'$
\EndIf
\EndFor
%\If{$|L_1| \ge |L_2|$}
%\State \textbf{return} $(L_1, R_1)$
%\Else{ \textbf{return} $(L_2, R_2)$}
%\EndIf
\State \Return $L,R$
%\EndProcedure
\end{algorithmic}
\end{algorithm}

Before delving into the proof, we need a certain subroutine that ``extracts'' a large balanced biclique of a graph $G$ when given a subset $S$ of vertices of some (larger) balanced biclique of $G$.
This is the subroutine \DenseBicliqueFinder\ appearing at line \ref{algo:biclique:line:dense} of \cref{alg:mainbiclique}, and it plays a role similar to the one played by \DenseCliqueFinder\ in the case of clique.

%Our analysis capitalizes on the fact that, given a subset $S$ of vertices that are contiguous in the ordering and that belong to the randomly planted balanced biclique, we can find in polynomial time two sides of a biclique both of size $\Omega (|S|)$.

%\marco{Non sono sicuro di capire la frase qui sopra}
%\francesco{Era imprecisa: credo di aver sistemato}

\begin{lemma}\label{lem:bicliquepoly}
There exists a deterministic algorithm $\scA$ with the following guarantees.
Let $G=(V,E)$ be an $n$-vertex graph containing a balanced biclique with sides $A$ and $B$. Given in input $G$ and $S \subseteq A \cup B$, algorithm $\scA$ returns a biclique of $G$ with sides $L, R$ such that $\min(|L|,|R|) \ge \frac{|S|}{3}$. The running time of $\scA$ is $O ( |S| \cdot n \log n)$.
\end{lemma}
\begin{proof}
We prove that \cref{alg:biclique_subroutine} satisfies the statement.

\begin{algorithm}[H]
\caption{$\DenseBicliqueFinder(G,S)$}\label{alg:biclique_subroutine}
\begin{algorithmic}[1]
\State compute the complement $\bar G[S]$ of $G[S]$
\State compute $G_1,\ldots,G_r$, the connected components of $\bar G[S]$ in nonincreasing order of vertex size 
\If{$|V(G_1)| > \frac{|S|}{3}$}
\State \Return $L = V(G_1)$ and $R = \cap_{u \in V(G_1)} N_G(u)$
\Else
\State compute the smallest $i \in [r]$ such that $\sum_{j=1}^i |V(G_i)| > \frac{|S|}{3}$
\State \Return $L = \bigcup_{j=1}^i V(G_j)$ and $R = \bigcup_{j=i+1}^r V(G_j)$
\EndIf
\end{algorithmic}
\end{algorithm}

%Let $\bar G[S]$ be the complement of $G[S]$ and let $G_1,\ldots,G_r$ be its connected components in non-increasing order of cardinality, as computed by the algorithm.
First, let us prove that the algorithm returns sets $L, R$ that are sides of a complete biclique and such that $\min(|L|,|R|)>\frac{|S|}{3}$. 
We begin by noting the following crucial fact: for each $i=1,\ldots,r$ we have $V(G_i) \subseteq A$ or $V(G_i) \subseteq B$.
Suppose, in fact, that there exist $u \in V(G_i) \cap A$ and $v \in V(G_i) \cap B$.
Since $G_i$ is connected, along any path from $u$ to $v$ in $G_i$ there must exist an edge whose vertices belong to $A$ and $B$, respectively.
Without loss of generality we can thus assume that $u$ and $v$ are such vertices.
By definition of $G_i$ this means that $u$ and $v$ are not adjacent in $G$, a contradiction.
Now we distinguish the two cases on which the algorithm branches.
\begin{enumerate}\itemsep0pt
    \item $|V(G_1)| > \frac{|S|}{3}$.
    %In this case, we compute the intersection of the neighborhoods of the vertices of $G_1$ in $G$, that is, $W=\bigcap_{u \in V(G_1)}N(u)$, and return the biclique with sides $L=V(G_1)$ and $R=W$. (Note that the neighborhoods are taken in $G$, not in $\bar G$ or $\bar G[S]$).
    In this case, as $V(G_1) \subseteq A$ or $V(G_1) \subseteq B$, by construction of $R$ we have $R \supseteq B$ or $R \supseteq A$. Therefore:
    \begin{align}
        |R| \ge |A| = \frac{|A \cup B|}{2} \ge \frac{|S|}{2} \ge \frac{|S|}{3}
    \end{align}
    Hence, $\min(|L|,|R|)>\frac{|S|}{3}$.
    Moreover note that $L, R$ are sides of a biclique by construction of $R$.
    
    \item $|V(G_1)| \le \frac{|S|}{3}$.
    Then by the ordering of $G_1,\ldots,G_r$ we have $|V(G_i)| \le \frac{|S|}{3}$ for all $i=1,\ldots,r$.
    Note how this implies that the index $i$ computed by the algorithm satisfies:
    %we compute the smallest $i$ such that:
    %\begin{align}
    %    \left| \bigcup_{j=1}^i V(G_j) \right| > \frac{|S|}{3}
    %\end{align}
    %This implies:
    \begin{align}
        \left| \bigcup_{j=1}^{i-1} V(G_j) \right| \le \frac{|S|}{3}
    \end{align}
    Therefore:
    \begin{align}
        \frac{|S|}{3} < \left| \bigcup_{j=1}^i V(G_j) \right| = \left| \bigcup_{j=1}^{i-1} V(G_j) \right| + \left| V(G_i) \right| \le \frac{|S|}{3} + \frac{|S|}{3} = \frac{2|S|}{3}
    \end{align}
%    We then return the biclique with sides:
%    \begin{align}
%        L = \bigcup_{j=1}^i V(G_j), \quad R = \bigcup_{j=i+1}^r V(G_j)
%    \end{align}
    This implies again $\min(|L|,|R|) \ge \frac{|S|}{3}$.
    Moreover $L, R$ form again the sides of a biclique; this is because $G_1,\ldots,G_r$ are connected components of $\bar G[S]$, hence in $G$ all edges are present between $V(G_j)$ and $V(G_{j'})$ for every distinct $j,j'$.
\end{enumerate}

%\marco{Ho cambiato un po'.}

%To summarize, our algorithm works in the following way. It sorts the connected components of $\overline{G}[S]$ in non-increasing order based on the cardinality of their vertex subsets: if the connected component with the highest number of vertices has more than $\frac{|S|}{3}$, it computes the intersection of all their neighborhoods in $G$ and returns the resulting biclique. Otherwise, it computes the minimum $i$ such that the union of the vertex sets of the first $i$ connected components has cardinality greater than $\frac{|S|}{3}$ and then returns the subsets $L$ and $R$, which contain respectively the vertices of the first $i$ connected components and all the others.

We now analyze the running time of the algorithm.
Computing $\bar G[S]$ takes time $O(|S|^2 \log n)$ by checking for each of the edges in the sorted adjacency lists of $G$.
Computing and sorting the connected components $G_{1}, \ldots , G_{r}$ takes time $O(|S|^2 + |S| \log |S|)$.
The case $|V ( G_{1} ) | > \frac{|S|}{3}$ requires time $O( |V(G_1)| \cdot n )= O ( |S| \cdot n )$ if the intersection of the neighborhoods is done using a bitmap indexed by $V(G)$.
The case $|V ( G_{1} ) | \le \frac{|S|}{3}$ takes time $O(|S|)$.
We conclude that the algorithm runs in time $O(|S|^2 \log n + |S| \cdot n) = O ( |S| \cdot n \log n )$.
\end{proof}

We are now ready to prove \cref{thm:main-biclique-formal}.

\begin{proof}[Proof of \cref{thm:main-biclique-formal}.]
We prove the biclique size guarantees and the running time bounds separately.

\textbf{Guarantees.}
Let $\beta=\frac{c}{2}$, and define:
\begin{align}
    f(\alpha,n,s) := \frac{n}{\alpha^{2 + \frac{1}{\beta} \log \frac{n}{s}} }
\end{align}
We begin by showing that, whenever $s \le \frac{n}{f(\alpha, n,s)}$ and $\alpha \ge \max(2,f(\alpha,n,s))$,
\begin{enumerate}
    \item If item (i) of \Cref{lem:structural} holds, then the first phase of \cref{alg:mainclique} finds a biclique with at least $\lfloor\frac{2}{3} \cdot f(\alpha,n,s)\rfloor$ vertices per side.
    \item If item (ii) of \Cref{lem:structural} holds, then the second phase of \cref{alg:mainclique} finds with high probability a biclique with at least $\frac{c}{6} \cdot f(\alpha,n,s)$ vertices per side.
\end{enumerate}
Since by \Cref{lem:structural} itself at least one of items (i) and (ii) holds, the algorithm finds with high probability a balanced biclique of size $\Omega(c \cdot f(\alpha,n,s))$.
We will then choose $\alpha,s$ that satisfy the constraints above while (roughly) maximizing $f(\alpha,n,s)$.

\ignore{
We will thus solve for the largest possibile $s$ satisfying $s < \frac{n}{f(\alpha,n,s)}$, and then $\frac{c}{3} \cdot f(\alpha,n,s)$ will be a lower bound on the size of the biclique found. The smaller the value of $\alpha$ we choose, the larger the lower bound $f(\alpha,n,s)$ we get; however, as will be clear in the proof, in order for $f(\alpha,n,s)$ to actually be a lower bound, the second item in the statement of \Cref{lem:good_set_densifies} forces that 
\begin{equation}\label{vincolo_alpha}
\alpha \ge f(\alpha,n,s).    
\end{equation}
%\medskip
%\marco{Try expanding this!}
%\medskip
}

We prove 1.
To ease the notation let $f=f(\alpha,n,s)$.
If item (i) of \cref{lem:structural} holds, then $|V_{<s}| \ge f$, hence in $G$ (and thus in $\hat G$) there are at least $f$ vertices of slack smaller than $s \le \nicefrac{n}{f}$.
By a simple counting argument, any $k$ of those vertices have at  least $n-k(s+1)$ neighbors in common.
Choosing $k = \lfloor \frac{2f}{3}\rfloor$, and using the fact that $s \le \nicefrac{n}{f}$, the common neighbors are at least $n - \frac{2f}{3}\left(\nicefrac{n}{f}+1\right) = \frac{n-2f}{3}$.
Now observe that, since $\alpha \ge 2$, then $n \ge 4f$, hence $\frac{n-2f}{3} \ge \frac{2f}{3}$.
We conclude that the loop at line \ref{algo:biclique:line:P1} of \cref{alg:mainbiclique} eventually returns a biclique whose smallest side has at least $\lfloor\frac{2f}{3}\rfloor$ vertices.

We prove 2.
If item (ii) of \cref{lem:structural} holds, then $G$ contains an $(\alpha, \nicefrac{c}{2})$-bulging set $U$ of size at least $\alpha f$.
Let $S = \hat V_{<s_U}$.
Leveraging \cref{lem:good_set_densifies} through the same arguments used in the proof of \Cref{thm:main-clique-formal}, as long as $\alpha f$ and $s$ are in $\Omega(c^{-1} a \log n)$ and sufficiently large, with probability at least $1-n^{-a}$ we have $|S \cap H| > \frac{c}{2}|U|$ and $|S \cap H| > \frac{c\alpha}{2}|S \setminus H|$, where $H$ is the set of vertices of the planted biclique.
Now consider any ordering of $S$ (in particular the one given by the degrees).
If $|S \setminus H| = \emptyset$, then the ordering itself is a sequence of elements of $H$ of length $|S|=|S \cap H| > \frac{c}{2}|U| \ge \frac{c\alpha f}{2}$.
If instead $|S \setminus H| \ne \emptyset$, as $|S \cap H| > \frac{c\alpha}{2}|S \setminus H|$ the pigeonhole principle implies that the ordering contains a contiguous sequence of vertices of $H$ of length at least $\frac{c\alpha}{2}$, and therefore are least $\frac{c f}{2}$ as we are assuming $\alpha \ge f$.
We conclude that in any case the ordering of $S$ contains a contiguous sequence of vertices of $H$ of length $\min\Paren{\frac{c f}{2},\frac{c \alpha f}{2}} = \frac{c f}{2}$.
By construction \cref{alg:mainbiclique} eventually runs \DenseBicliqueFinder\ on that sequence and thus, by \cref{lem:bicliquepoly}, finds a biclique with at least $\frac{c f}{6}$ vertices per side.

It remains to choose suitable values of $\alpha,s$ so as to approximately maximize $f$ subject to the constraints $s \le \frac{n}{f(\alpha, n,s)}$ and $\alpha \ge \max(2,f(\alpha,n,s))$.
The argument above then yields with probability $1-n^{-a}$ a biclique with $\Omega(c f(\alpha,n,s))$ vertices per side.
We set:
\begin{align}
    s &= \frac{n}{f} \\
    \alpha &= f
\end{align}
This yields the equation:
\begin{equation}
\alpha = f(\alpha,n,s) = \frac{n}{\alpha^{2 + \frac{1}{\beta} \log \frac{n}{s}}} = \frac{n}{\alpha^{2 + \frac{1}{\beta} \log \alpha}}
\end{equation}
Recalling that $\beta=\nicefrac{c}{2}$, rearranging, and taking logarithms yields: 
\begin{equation} 
\frac{2}{c} \log^2 \alpha + 3 \log \alpha - \log n = 0.
\end{equation}
Solving for $\log \alpha$ gives:
\begin{align} 
\log \alpha = \frac{- 3 + \sqrt{9 + \frac{8 \log n}{c}} }{\frac{4}{c}} = \sqrt{ \frac{9}{16} c^2 + \frac c 2 \log n} -\frac{3}{4}c > \sqrt{\frac c 2 \log n} - 1
\end{align}
We conclude that:
\begin{equation}\label{bound_alpha}
    \alpha > 2^{\sqrt{\frac c 2 \log n} - 1}
\end{equation}
Notice that by definition $s,\alpha$ satisfy the constraint $\alpha \ge \max(2,f(\alpha,n,s))$ as long as $\alpha = f \ge 2$. Since we are assuming that $c \in \omega (\frac{1}{\log n})$, \cref{bound_alpha} guarantees that $f = \alpha \ge 2$ holds for large enough $n$.

Finally, the lower bound above on the size of each side of the biclique is thus:
\begin{align}
    \frac{cf}{6} = \frac{c\alpha}{6} > \frac{c}{12} \cdot 2^{\sqrt{\frac c 2 \log n}}
\end{align}

\textbf{Running time.}
We describe a variant of \BalancedBicliqueFinder\ that runs in time $\Tilde{O}(\lVert \hat G\rVert)$ and finds a biclique of size at least $\nicefrac{1}{4}$-th of that of the original algorithm.
As a first thing, we compute $|E(\hat G)|$. If $|E(\hat G)|<\Paren{\frac{cn}{2}}^2$, then necessarily $c \le \frac{1}{\log n}$, and the bound of \cref{thm:main-biclique-formal} is smaller than $1$.
In this case we immediately return any edge of $G$, satisfying the bounds.
If instead $|E(\hat G)|\ge \Paren{\frac{cn}{2}}^2$ then we run the $\tilde{O}(n^2)$-time variant of \cref{alg:mainbiclique} described below.
This makes the running time in $\Tilde{O}(\lVert \hat G\rVert)$ in every case.
As a byproduct, the lower bound on the biclique size will shrink by a factor $\frac{4}{5}$.

The variant of \cref{alg:mainbiclique} is as follows.
First, observe that the loop of line \ref{algo:biclique:line:P1} can be implemented in $\tilde{O}(n^2)$ total time by computing $R'$ incrementally (this can be done either via a bitmap or using binary search over the sorted adjacency lists).
For the loop at line \ref{algo:biclique:line:P2}, we reduce the running time by coarsening.
Instead of iterating over all $1 \le i < j \le n$, for each $h=1,\ldots,\lceil\log n\rceil$ we iterate over all subsequences $v_i,\ldots,v_j$ with $i=k2^h$ and $j=k2^h + k-1$, for $k=0,1,2,\ldots$.
Clearly, for every contiguous subsequence $S$ of $v_1,\ldots,v_n$, we will iterate over some subsequence $S' \subseteq S$ with $|S'| \ge |S|/4$.
The bound on the size of the biclique thus decreases by a factor of $4$.
The running time can be easily bounded by noting that, for every $h=1,2,\ldots$, the total cost of invoking \DenseBicliqueFinder\ on all the subsequences of size $2^i$ is in $\tilde{O}(n^2)$ by \cref{lem:bicliquepoly}.
As the loop iterates over $O(\log n)$ values of $i$, we conclude that the second phase takes $\tilde{O}(n^2)$ time overall.

\end{proof}

\section{A lower bound on densification}\label{sec:lower-bound}
%\francesco{Attenzione: questa sezione continua ad usare la notazione $\kappa$ che avevamo smesso di usare; quindi direi che o usiamo una notazione alternativa oppure la si reintroduce nella notazione (nei preliminaries credo)}
%\tom{@Marco, check this phrase}
In this section we prove  \cref{thm:densification_LB}. 
This shows that, whenever $c < 1/2$, there exist arbitrarily large graphs $G$ such that the high degree profiles of typical instances from $\cG(G,K_{cn})$ are essentially uncorrelated with the planted clique.  
%\begin{theorem}\label{thm:densification_LB}
%For every $c \in (0,\frac{1}{2})$ and $n \ge 3$ there exists an $n$-vertex graph $G$ such that $\hat G \sim \scG(G,K_{cn})$ satisfies what follows with probability at least $1-\frac{1}{n}$.
%Let $\hat G \sim \scG(G,K_{cn})$ and let $\hat V_{\le s}$ be the set of vertices with slack at most $s$ in $\hat G$.
%Then, with probability at least $1-\frac{1}{n}$ over $\hat G$, for all $s \ge 0$:
%For every ordering $v_1,\ldots,v_n$ of the vertices of $\hat G$ by nonincreasing degree, and for every $j\in [n]$, the largest clique in the induced subgraph $\hat G\brac{\set{v_1,\ldots,v_j}}$ has size at most
%\begin{align}
    %\big|\hat V_{\le s}\big| \ge b \sqrt{\frac{n \ln n}{c}} \quad \Rightarrow \quad \kappa\left(\hat G\big[\hat V_{\le s}\big]\right) = O\!\left(c \cdot \big|\hat V_{\le s}\big|\right)\,.
    %\kappa\left(\hat G\big[\{v_1,\ldots,v_j\}\big]\right) \le
%    O\left( \frac{\sqrt{n \ln n}}{c} + c \, j\right)\,.
%    \label{eq:LB_clique_bound}
%\end{align}
%where $b>0$ is a universal constant independent of $c$, $n$ and $j$.
%\end{theorem}

Throughout the section, for a graph $G$ we let $\kappa(G)$ be the size of the largest clique in $G\,.$
We start by defining a graph $H$ that has between one and two vertex for every degree (or, equivalently, every slack) from $1$ to $n-1$.
Let $H=(V,E)$ where $V=[n]$ for $n \ge 3$, and
\begin{align}
E=\Big\{\{u,v\} : u,v \in V, u \ne v, u+v \le n+1\Big\} \, .  
\end{align}
Note that $N_H(u) = [1,n-u+1] \setminus \{u\}$; hence
\begin{align}
    s_u = \left\{
    \begin{array}{ll}
        u-1 & u \le \frac{n+1}{2} \\
        u-2 & u > \frac{n+1}{2}
    \end{array}
    \right.
\end{align}
This implies that, for every $0 \le s \le n-1$,
\begin{align}
    V_{\le s} \in [s+1, s+2] \,.
\end{align}

\noindent The graph $G$ of \Cref{thm:densification_LB} is a perturbation of $H$ as given by the next result.
\begin{lemma}\label{lem:random_LB_graph}
    Let $G$ be an $n$-vertex graph, let $\eta \in [0,1]$, and let $G'$ be obtained from $G$ by deleting each edge independently with probability $\eta$.
    For every $a > 1$, with probability at least $1-2n^{1-a}$:
    \begin{enumerate}\itemsep0pt
        \item $\kappa(G') < 2a\frac{\ln n}{\eta} + 1$.
        \item $|V_{\le s}| \le |V'_{\le s'}|$ for all $s \ge 0$, where $s' = s + \eta n + \sqrt{a n \ln n}$.
    \end{enumerate}
\end{lemma}
\begin{proof}
\emph{Item 1.} Fix $U \subseteq V$ on $k \ge 2a\frac{\ln n}{\eta} + 1$ vertices. Then:
\begin{align}
    \Pr[G'[U] \text{ is a clique}]
    \le \left(1-\eta\right)^{\binom{k}{2}}
    < e^{-\eta \binom{k}{2}}
    = e^{-k \cdot \eta \frac{k-1}{2}}
    \le e^{-k a \ln n}
    = n^{-a k} \,.
\end{align}
Taking a union bound over all $U$ yields $\Pr[\kappa(G') \ge k] < n^{(1-a)k} \le n^{1-a}$.

\emph{Item 2.} 
Fix $u \in V$. Then $s'_u = s_u + \sum_{i=1}^{d_u}X_i$, where the $X_i$ are independent Bernoulli random variables with parameter $\eta$. By Hoeffding's inequality, for every $t \ge 0$,
\begin{align}
    \Pr[s'_u \ge s_u + \eta n + t] \le \Pr[s'_u \ge s_u + \eta d_u + t] \le e^{-\frac{t^2}{d_u}} < e^{-\frac{t^2}{n}}
\end{align}
For $t = \sqrt{a n\ln n}$ we obtain $\Pr[s'_u \ge s_u + \eta n + \sqrt{a n \ln n}]\le n^{-a}$.
This implies that, for every $s \ge 0$, every $v \in V_{\le s}$ satisfies $v \in V'_{\le s'}$ with probability $1-n^{-a}$, where $s' = s + \eta n + \sqrt{a n \ln n}$.
By a union bound we conclude that, with probability $1-n^{1-a}$, we have $|V'_{\le s'}| \ge |V_{\le s}|$ for all $s \ge 0$.
\end{proof}

As a corollary we get the graph $G$ used in the proof of \Cref{thm:densification_LB}:
\begin{corollary}\label{cor:exists_lb_graph}
For every $\eta \in [0,1]$ and every $n \ge 3$ there exists an $n$-vertex graph $G$ such that:
\begin{enumerate}\itemsep0pt
    \item $\kappa(G) \le 4 \frac{\ln n}{\eta} + 1$.
    \item $s - \eta n - \sqrt{2 n \ln n} \le |V_{\le s}| \le s+2$ for all $s \ge 0$.
\end{enumerate}
\end{corollary}
\begin{proof}
    Apply \Cref{lem:random_LB_graph} to the graph $H$ defined above for $a=2$, noting that $1-2n^{1-a}>0$.
    %Together with the lemma above, and $1-2n^{1-C}>0$, this proves the claim.
\end{proof}

The next result bounds the number of vertices of the planted clique that end up having a certain slack in $\hat G$.
\begin{lemma}\label{lem:KcapV_lb}
Let $c\in (0,1)$, let $G$ be any graph, and let $\hat G \sim \scG(G,K_{cn})$.
With probability at least $1-\frac{3}{n}$ we have simultaneously for all $s \ge 0$:
\begin{align}
    c \cdot |V_{\le s}| - \sqrt{2 n \ln n} \le \big| K \cap \hat V_{\le s} \big| \le c \cdot |V_{\le s^*}| + \sqrt{2 n \ln n}.
\end{align}
where $s^* = \frac{s + \sqrt{2 n \ln n}}{1-c}$.
\end{lemma}
\begin{proof}
\emph{Lower bound.}
Note that $|K \cap \hat V_{\le s}| \ge |K \cap V_{\le s}|$, and $|K \cap V_{\le s}| = \sum_{i=1}^{|V_{\le s}|}$ where the $X_i$ are non-positively correlated Bernoulli random variables of parameter $c$.
By Hoeffding's inequality, then, the probability that the lower bound of the claim fails is at most $\frac{1}{n^2}$ for any given $s \ge 0$.
By a union bound, thus, the lower bound holds fails for some $s$ with probability at most $\frac{1}{n}$.

\emph{Upper bound.}
Let $v \notin V_{\le s^*}$, so $s_v > s^*$.
Note that $\hat s_v = s_v - \sum_{i=1}^{s_v} X_i$, with the $X_i$ non-positively correlated Bernoulli random variables of parameter $c$.
Therefore $\E[\hat s_v] = (1-c)s_v$, and:
    \begin{align}
        s = (1-c)s^* - \sqrt{2 n \ln2 n} < (1-c)s_v - \sqrt{2 n \ln n} = \E[\hat s_v] - \sqrt{2 n \ln 2n}
    \end{align}
    By Hoeffding's inequality we then get $\Pr\big[ \hat s_v \le s \big] \le \frac{1}{n^2}$.
    By a union bound this implies that, with probability at least $1-\frac{1}{n}$,
    \begin{align}
        K \cap \hat V_{\le s} \subseteq K \cap V_{\le s^*} \qquad \forall s=1,\ldots,n-1
    \end{align}
    Consider then $|K \cap V_{\le s^*}|$.
    Note that this is a sum of $|V_{\le s^*}|$ non-positively correlated Bernoulli random variables of parameter $c$.
    Another application of Hoeffding's inequality yields with probability at least $1-\frac{1}{n^2}$:
    \begin{align}
    \big|K \cap V_{\le s^*}\big| \le c \cdot |V_{\le s^*}| + \sqrt{2 n \ln n}
    \end{align}
    
    A final union bound over all $s\ge 0$ and the three events above concludes the proof.
\end{proof}

We are now ready to prove \Cref{thm:densification_LB}.

\begin{proof}[Proof of \Cref{thm:densification_LB}]
Let $\eta = a \cdot c^{-1}\sqrt{\frac{\ln n}{n}}$ for some $a>0$ to be defined.
Let $G$ be the corresponding graph given by \Cref{cor:exists_lb_graph}, and let $\hat G \sim \scG(G,K_{cn})$.
We begin by observing that it is sufficient to prove \Cref{thm:densification_LB} for the case $\{v_1,\ldots,v_j\} = \hat V_{\le s}$ for some $s \ge 0$.

Consider indeed any ordering $v_1,\ldots,v_n$ of the vertices of $\hat G$ by nonincreasing degree.
Observe that for every $j=1,\ldots,n$ there exists $s \ge 0$ and $\hat S \subseteq \hat V_{\le s} \setminus \hat V_{\le s-1}$ such that
\begin{align}
    \{v_1,\ldots,v_j\} = \hat V_{\le s-1} \,\dot\cup\, \hat S
\end{align}
Now suppose the bound of \Cref{lem:KcapV_lb} holds.
We claim that $|\hat S| \le (2+a) \frac{\sqrt{n \ln n}}{c}$.
Indeed:
\begin{align}
|\hat S| &
\le |\hat V_{\le s}| \setminus |\hat V_{\le s-1}| \\
&
\le |\hat V_{\le s}| \setminus |V_{\le s-1}| && V_{\le s-1} \subseteq \hat V_{\le s-1} \\
& \le \left(c \cdot\frac{s + \sqrt{2 n \ln n}}{1-c} + \sqrt{2n \ln n}\right) - \big(s-1 - \eta n \big) && \text{\Cref{lem:KcapV_lb} and \Cref{cor:exists_lb_graph}} \\
& = \left(c \cdot\frac{s + \sqrt{2 n \ln n}}{1-c} + \sqrt{2n \ln n}\right) - \big(s-1 - a \frac{\sqrt{n \ln n}}{c} \big) && \text{definition of }\eta \\
&\le (2+a) \frac{\sqrt{2 n \ln n}}{c}
\end{align}
where in the last inequality we used $c \le \frac{1}{2}$.
Now notice that the upper bound of \Cref{thm:densification_LB} has an $O\left(\frac{\sqrt{2 n \ln n}}{c}\right)$ additive term.
Therefore, as said, it is sufficient to prove the theorem for the case $\{v_1,\ldots,v_j\} = \hat V_{\le s}$ for some $s \ge 0$.

Consider then any $0 \le s \le n-1$.
If $\big|\hat V_{\le s}\big| \le a \cdot c^{-1}\sqrt{n \ln n}$ then \Cref{eq:LB_clique_bound} is trivially true.
Suppose then $\big|\hat V_{\le s}\big| > a \cdot c^{-1} \sqrt{n \ln n}$.
We have: 
\begin{align}
    \big| K \cap \hat V_{\le s} \big| &\le c \cdot |V_{\le s^*}| + \sqrt{2 n \ln n} && \text{\Cref{lem:KcapV_lb}} \\
    & \le c \cdot (s^* + 2) + \sqrt{2 n \ln n} && \text{item 2 of \Cref{cor:exists_lb_graph}} \\
    &= O\left(cs + \sqrt{n \ln n}\right) && \text{definition of } s^* \text{ and }c \le \frac{1}{2}
    %&= O\left(c \cdot \big|\hat V_{\le s}\big| + \sqrt{n \ln n}\right) && \text{definition of } s^* \text{ and }c \le \frac{1}{2}
\end{align}
By item 1 of \Cref{cor:exists_lb_graph}, and since $V_{\le s} \subseteq \hat V_{\le s}$, we have $s \le \big|\hat V_{\le s}\big|$.
As $\big|\hat V_{\le s}\big| > a \cdot \sqrt{c^{-1}\, n \ln n}$, we have $\sqrt{n \ln n} < \frac{c}{a} \big|\hat V_{\le s}\big|$.
Plugging these bounds in the inequality above gives $\big| K \cap \hat V_{\le s} \big| = O(c |\hat V_{\le s}\big|)$.
To conclude, observe that:
\begin{align}
    \kappa\left(\hat G\big[\hat V_{\le s}\big]\right) \le \kappa(G) + \big| K \cap \hat V_{\le s} \big|
\end{align}
and that $\kappa(G) \le \frac{\ln n}{\eta} = a \, c \sqrt{n \ln n}$ by \Cref{cor:exists_lb_graph} and our choice of $\eta$.
Together with our bound on $\big| K \cap \hat V_{\le s} \big|$ this gives the claim.
\end{proof}

\ifnum\draft=1
\clearpage
\section*{TODOs}
\begin{enumerate}
\item semplificare dimostrazione del lemma bilanciamento biclique
\item riformulare nostro risultato come riduzione al problema worst-case con $c=\Omega(1)$, per poter fare plug-in di algoritmi tipo Manurangsi+AlonKahale
\item  controllare la definizione di slack critico ($m/n$ ?) 
\item  sostituire $n/s(G_0)$ con una variabile sensata
\item generalizzare al caso $\Delta(G_0)\le n - s(n)$; ancora meglio in funzione del numero $n_s$ di vertici con slack $> s$.
\item capire quali proprietà della distribuzione servono davvero
\item notare robustezza rispetto a cambiamenti monotoni (es.\ rimozione archi fuori da $K$)
\item studiare grafo avversariale perturbato con probabilità $p$. Riusciamo a interpolare fra $G(n,p)$ e caso peggiore?
\item LB per teorema densificazione (esiste $G$ tale che in $\hat G$ nessuna sottosequenza di vertici ordinati contiene una clique di taglia $>$ BLA).
\item raffinare analisi fasce usando gradi precisi più deviazione standard 
/ alta probabilità. Com'è il grafo cattivo?
\item dare un running time concreto all'algoritmo (es.\ $n^5$): funziona in tempo $\tilde{O} (|G|)?$
\item provare BH al posto di AK
\item inserire osservazioni aggiuntive (es.\ sui core)
\item che succede se invece di piantare clique/biclique piantiamo un qualunque sottografo denso? Cosa recuperiamo?
\item si può fare qualcosa per $t$-clique piantata?
\item analisi algoritmo pelatura greedy, almeno su qualche esempio informativo
\item tecniche di rimozione di sottografi (in che modo sarebbero utili?)
\item riduzione da clique a colored partitioned biclique
\end{enumerate}
\fi

\phantomsection
\addcontentsline{toc}{section}{Bibliography}
{\footnotesize
\bibliographystyle{amsalpha} 
\bibliography{scholar, tesiagrimonti}
}

\clearpage
\appendix
\section{Concentration inequalities}\label{sec:concentration}
%We review some Chernoff-type probability bounds that are repeatedly used in our analysis. 
The following bounds can be found in \cite{augerdoerr} or derived from \cite{panconesidub}.
Let $X_1, \ldots, X_n$ be binary random variables.
We say that $X_1, \ldots, X_n$ are non-positively correlated if for all $I \subseteq \{ 1, \ldots , n \}$:
\begin{equation}
    \Pr \left( \forall i \in I \mid X_i = 0 \right) \le \prod_{i \in I} \Pr \left( X_i = 0 \right)
\end{equation}
and
\begin{equation}
\Pr \left( \forall i \in I \mid X_i = 1 \right) \le \prod_{i \in I} \Pr \left( X_i = 1 \right).
\end{equation}
Then:
\begin{lemma}\label{lem:chernoff}
Let $X_1, \ldots, X_n$ be independent or, more generally, non-positively correlated binary random variables. Let $a_1, \ldots, a_n \in \left[ 0, 1 \right]$ and $X = \sum_{i = 1}^{n} a_i X_i$. Then, for any $\varepsilon > 0$, we have:
\begin{equation}
\Pr \left( X \le (1 - \varepsilon) \E \left[ X \right] \right) \le  e^{- \frac{\varepsilon^2}{2} \E \left[ X \right]}
\end{equation}
and
\begin{equation}
\Pr \left( X \ge (1 + \varepsilon) \E \left[ X \right] \right) \le  e^{- \frac{\varepsilon^2}{2 + \varepsilon} \E \left[ X \right]}
\end{equation}
\end{lemma}
%Note that~\Cref{lem:chernoff} applies if $X_1, \ldots, X_n$ are indicator variables of mutually disjoint events, or can be partitioned into independent families of such variables.

\end{document}